\pdfoutput=1
\RequirePackage{ifpdf}
\ifpdf % We are running pdfTeX in pdf mode
\documentclass[pdftex]{sigma}
\else
\documentclass{sigma}
\fi

\numberwithin{equation}{section}

\newtheorem{thrm}{Theorem}[section]
\newtheorem{prop}[thrm]{Proposition}
\newtheorem{cor}[thrm]{Corollary}

\theoremstyle{definition}
\newtheorem{Definition}[thrm]{Definition}
\newtheorem{rem}[thrm]{Remark}

\begin{document}

\renewcommand{\thefootnote}{$\star$}

\newcommand{\arXivNumber}{1512.02386}

\renewcommand{\PaperNumber}{087}

\FirstPageHeading

\ShortArticleName{B\"acklund Transformations and Non-Abelian Nonlinear Evolution Equations}

\ArticleName{B\"acklund Transformations and Non-Abelian\\ Nonlinear Evolution Equations:\\ a Novel B\"acklund Chart\footnote{This paper is a~contribution to the Special Issue on Analytical Mechanics and Dif\/ferential Geometry in honour of Sergio Benenti.
The full collection is available at \href{http://www.emis.de/journals/SIGMA/Benenti.html}{http://www.emis.de/journals/SIGMA/Benenti.html}}}

\Author{Sandra CARILLO~$^{\dag^1\dag^2}$, Mauro LO SCHIAVO~$^{\dag^1}$ and Cornelia SCHIEBOLD~$^{\dag^3\dag^4}$}

\AuthorNameForHeading{S.~Carillo, M.~Lo Schiavo and C.~Schiebold}

\Address{$^{\dag^1}$~Dipartimento ``Scienze di Base e Applicate per l'Ingegneria'', Sapienza - Universit\`a di Roma,\\
\hphantom{$^{\dag^1}$}~16, Via A.~Scarpa, 00161 Rome, Italy}
\EmailDD{\href{mailto:sandra.carillo@sbai.uniroma1.it}{sandra.carillo@sbai.uniroma1.it}, \href{mailto:mauro.loschiavo@sbai.uniroma1.it}{mauro.loschiavo@sbai.uniroma1.it}}
\URLaddressDD{\url{http://www.sbai.uniroma1.it/~sandra.carillo/}}
\URLaddressDD{\url{http://www.sbai.uniroma1.it/~mauro.loschiavo/}}

\Address{$^{\dag^2}$~I.N.F.N. - Sez. Roma1, Gr.~IV - Mathematical Methods in NonLinear Physics, Rome, Italy}

\Address{$^{\dag^3}$~Department of Science Education and Mathematics, Mid Sweden University,\\
\hphantom{$^{\dag^3}$}~S-851 70 Sundsvall, Sweden}
\EmailDD{\href{mailto:Cornelia.Schiebold@miun.se}{Cornelia.Schiebold@miun.se}}
\URLaddressDD{\url{http://www.miun.se/personal/corneliaschiebold}}

\Address{$^{\dag^4}$~Instytut Matematyki, Uniwersytet Jana Kochanowskiego w Kielcach, Poland}

\ArticleDates{Received December 08, 2015, in f\/inal form August 24, 2016; Published online August 30, 2016}

\Abstract{Classes of third order non-Abelian evolution equations linked to that of Kor\-te\-weg--de Vries-type are investigated and their connections represented in a non-commutative B\"acklund chart, generalizing results in [Fuchssteiner B., Carillo S., \textit{Phys.~A} \textbf{154} (1989), 467--510]. The recursion operators are shown to be hereditary, thereby allowing the results to be extended to hierarchies. The present study is devoted to operator nonlinear evolution equations: general results are presented. The implied applications referring to f\/inite-dimensional cases will be considered separately.}

\vspace{3mm}

\rightline{\it Dedicated to Sergio Benenti on his 70th birthday}

\renewcommand{\thefootnote}{\arabic{footnote}}
\setcounter{footnote}{0}

\vspace{-3.5mm}

\section{Introduction}

B\"acklund transformations and their associated nonlinear superposition principles, known as permutability theorems, play a central role in modern soliton theory as described in \cite{Gu-book, RogersSchief, RogersShadwick}, see also \cite{RogersAmes} for further applications. The importance of the combination of gauge and reciprocal transformations to link hierarchies of integrable systems is well-established (see, e.g., \cite{OevelRogers} and literature cited therein). Here, the concern is with the application of B\"acklund transformations to connect various non-Abelian equations to a~canonical KdV-type equation and to construct their recursion operators. Hierar\-chies of such nonlinear equations are considered wherein the unknown is an operator on a~Banach space. Such operator equations were originally introduced by Marchenko \cite{Marchenko} and further investigated and developed, in the framework of Banach space geometry, in \cite{Aden-Carl, Carl-Schiebold, Carl-Schiebold-1}.

The results comprised in this article represent the continuation of the study in \cite{CMS-2016a, Carillo:Schiebold:JMP2009, SIMAI2008,Carillo:Schiebold:JMP2011, Carillo:Schiebold:JNMP2012}. In particular, in \cite{Carillo:Schiebold:JMP2009, Carillo:Schiebold:JMP2011} the focus is on the operator potential Korteweg--de Vries, Korteweg--de Vries and modif\/ied Korteweg--de Vries equations, their connections via B\"acklund transformations and the construction of the recursion operators they admit.

One of the main advantages of connecting non-Abelian equations via B\"acklund transformations is that solutions can be transferred from one equation to another. In \cite{Carillo:Schiebold:JMP2009} operator-valued solutions (which can be interpreted as operator analogs of solitons) to the pKdV, KdV and mKdV hierarchies are constructed. In \cite{Carillo:Schiebold:JMP2011} suitable projection techniques are exploited to derive solution formulae to the corresponding scalar and matrix hierarchies. Note that the study of non-Abelian nonlinear evolution equations found its original interest in the case of matrix equations \cite{CalogeroDegasperis, LeviRB}. The results in the present article are valid for operator-valued functions. This level of generality permits to construct solutions to scalar and matrix equations that can be viewed as countable superposition of solitons, see also \cite{Schiebold-1, Schiebold-6dic3} for the connection between countable nonlinear superposition and Banach space geometry, and \cite{Schiebold-6dic2, Schiebold-6dic1} for further applications.

In \cite{Rogers:Carillo:1987}, the term {\it B\"acklund chart} was introduced to indicate the net of B\"acklund transformations connecting dif\/ferent evolution equations. In \cite{Fuchssteiner:Carillo:1989a} a wide B\"acklund chart which includes {\it scalar} 3rd order Abelian evolution equations is constructed. It connects, further to the KdV and the modif\/ied KdV equations, in particular the KdV singularity manifold equation, also known as as UrKdV or Schwarz-KdV \cite{Depireux, Schiff, Weiss, Wilson}, and the KdV interacting soliton equation \cite{Fuchssteiner1987}. The connections among the equations in this Abelian B\"acklund chart are applied to f\/ind, or recover, the recursion operators admitted by all these nonlinear evolution equations. Moreover, other structural properties such as the Hamiltonian, and bi-Hamiltonian, structure of these equations are also obtained via the B\"acklund chart \cite{Fuchssteiner:Carillo:1989a}. The present study is concerned about the extension of the {B\"acklund chart} in \cite{Carillo:Schiebold:JMP2009} to obtain the non-Abelian analog of the links established in the scalar case \cite{Fuchssteiner:Carillo:1989a}. Precisely, in addition to the pKdV, KdV and mKdV equations already connected in \cite{Carillo:Schiebold:JMP2009}, a dif\/ferent version of the non-Abelian modif\/ied KdV and two non-Abelian equations, respectively, analogs of the KdV singularity manifold and of the {KdV interacting soliton} equation, are all linked together via B\"acklund transformations. New, in the present B\"acklund chart, are the inclusion of a second modif\/ied KdV equation, denoted as amKdV, for {\it alternative} mKdV equation \cite{KK, Olver:Sokolov}, new are also the non-Abelian {KdV singularity manifold} equation and the non-Abelian KdV interacting soliton equation.

All the non-Abelian nonlinear evolution equations in the B\"acklund chart admit a~recursion operator. The recursion operators of some of them, such as the pKdV, KdV and mKdV \cite{Carillo:Schiebold:JMP2009, Olver:Sokolov, Schiebold2010} are known. We recover, via the established connections, the recursion operators admitted by the non-Abelian amKdV, given in \cite{Olver:Sokolov}. Then, the recursion operators of the non-Abelian KdV singularity manifold and of the KdV interacting soliton equations, both new, are constructed. Furthermore, the hereditariness of all the obtained recursion operators is proved combining links via B\"acklund transformations \cite{FokasFuchssteiner:1981, Fuchssteiner1979}, with the hereditariness of non-Abelian KdV recursion operator~\cite{Schiebold2010}. Finally, since all the nonlinear third order non-Abelian evolution equations admit hereditary recursion operators, according to \cite{FokasFuchssteiner:1981, Fuchssteiner1979}, all the links in the B\"acklund chart can be extended to the corresponding whole hierarchies.

The material is organized as follows. The opening Section~\ref{sec4} is devoted to the B\"acklund charts connecting nonlinear evolution equations which generalize, to the operator level, the pKdV, KdV and mKdV equations \cite{Carillo:Schiebold:JMP2009,Carillo:Schiebold:JMP2011}. Sections~\ref{s mkdvs} and~\ref{s new} are devoted to the construction of a novel B\"acklund chart. In Section~\ref{s mkdvs}, the B\"acklund chart is extended to incorporate also the non-Abelian amKdV equation; its recursion operator is constructed from the known recursion operators of the mKdV equations; f\/inally, the connection between the two dif\/ferent non-Abelian modif\/ied KdV equations is provided.

In Section~\ref{s new}, the novel B\"acklund chart is further enlarged. Then, the recursion operators of the KdV interacting soliton and of the KdV singularity manifold equations are constructed. In addition, in Section~\ref{s new}, a M\"obius type invariance exhibited by the non-Abelian KdV singularity manifold equation is established.

The subsequent Section~\ref{section5} is devoted to the proof of hereditariness of all the recursion operators in the previous sections. The hereditariness of all the recursion operators in the B\"acklund chart guarantees that the links can be transferred to whole hierarchies, relating their corresponding members. In the closing Section~\ref{rems}, further to some remarks on the interest of the present study, open problems and perspectives suggested by our new results are brief\/ly outlined. The article is complemented with an Appendix where some needed def\/initions, such as the def\/inition of B\"acklund transformation~\cite{FokasFuchssteiner:1981}, adopted throughout, as well as known results obtained in the case of Abelian nonlinear evolution equations~\cite{Fuchssteiner:Carillo:1989a} are comprised.

\section{Non-Abelian B\"acklund charts}\label{sec4}

In \cite{Carillo:Schiebold:JMP2009}, the recursion operators admitted by the non-Abelian potential KdV, KdV and modif\/ied KdV equations were shown to be related by B\"acklund transformations. Consider f\/irst the operator potential KdV equation (pKdV)
\begin{gather}\label{pkdv}
W_t = W_{xxx}+{3} W_x^2,
\end{gather}
where the unknown $W$ is a function whose values are bounded linear operators on some Banach space\footnote{Here capital letters are used to emphasize that the unknown is an operator acting on a Banach space.}. It admits the recursion operator
\begin{gather} \label{pKdV recursion operator}
 \hat \Psi(W) = D^2 + A_{W_x} + D^{-1}A_{W_x}D + D^{-1}C_{W_{x}}D^{-1}C_{W_{x}},
\end{gather}
where $D$ denotes the derivative with respect to $x$, and $C_T$, $A_T$ denote the commutator and anti-commutator with respect to~$T$, namely,
\begin{gather*}
 C_T S:= [T,S] \equiv TS-ST, \qquad A_T S:= \{T,S\}\equiv TS+ST.
\end{gather*}
For earlier occurrences of recursion operators in the non-Abelian setting we refer to \cite{CalogeroDegasperis2,Fuchssteiner:Chowdhury,GKS,Olver:Sokolov}. Consider next the operator Korteweg--de Vries equation (KdV)
\begin{gather} \label{nc-kdv}
 U_t = U_{xxx}+3 \{U,U_{x}\},
\end{gather}
which admits the recursion operator
\begin{gather}\label{kdv-recop}
 \Phi(U) = D^2 + 2A_U + A_{U_x}D^{-1} + C_UD^{-1}C_UD^{-1}.
\end{gather}
The recursion operators (\ref{pKdV recursion operator}) and (\ref{kdv-recop}) are linked via the B\"acklund transformation
\begin{gather*}%\label{B1b}
B_1\colon \ U-W_x=0.
\end{gather*}
In fact, the transformation operator is $\Pi_{B_1}=D^{-1}$, and $\hat\Psi(W)=D^{-1}\Phi(U)D$. Hence, the non-Abelian pKdV (\ref{pkdv}) and KdV~(\ref{nc-kdv}) equations can also be written as
\begin{gather*}
W_t = \hat \Phi(W) W_x \qquad \text{and} \qquad U_t = \Phi(U) U_x.
\end{gather*}
Finally, the KdV equation (\ref{nc-kdv}) is related to the operator modif\/ied KdV equation (mKdV)
\begin{gather}\label{mkdv}
V_t = V_{xxx}-3 \big\{V^2 ,V_{x}\big\}
\end{gather}
via the Miura transformation
\begin{gather}\label{Mnc}
M\colon \ U + V_{x} + V^{2}= 0.
\end{gather}
Hence, the Miura transformation (\ref{Mnc}) allows to obtain the recursion operator $\Psi(V)$ of the mKdV (\ref{mkdv}) from the recursion operator $\Phi(U)$ of the KdV~(\ref{kdv-recop}) via
\begin{gather*}%\label{mkdv-recop}
\Psi(V)=\Pi_M \Phi(U) \Pi_M^{-1}, \qquad \text{where} \quad \Pi_M = -(D+A_V)^{-1}.
\end{gather*}
The latter implies
\begin{gather*}
 \Psi(V) = \big(D-C_VD^{-1}C_V\big)\big(D-A_VD^{-1}A_V\big)
 \end{gather*}
(see \cite{Carillo:Schiebold:JMP2009} for details), which can be equivalently written as
 \begin{gather}\label{mkdv-recop}
 \Psi(V) = (D-C_V)D^{-1}(D+C_V) (D-A_V)D^{-1}(D+A_V) .
\end{gather}
The following B\"acklund chart summarizes the links among the non-Abelian pKdV (\ref{pkdv}), KdV (\ref{nc-kdv}), and mKdV (\ref{mkdv}) equations:
 \begin{gather*}%\label{BCNH2}
\boxed{W_t = W_{xxx} + 3 W^2_x}\buildrel B_1 \over{\text{\textendash\textendash}}
\boxed{U_t = U_{xxx} + 3 \{ U,U_x\}}\buildrel M \over{\text{\textendash\textendash}}\boxed{V_t = V_{xxx}- 3 \big\{ V^2,V_{x}\big\}}\,.
\end{gather*}
When the respective recursion operators are applied to the above equations iteratively, the B\"acklund chart can be extended to the corresponding hierarchies, and the connections can be summarized in
\begin{gather}\label{BCNH3}
\boxed{W_{t} = [\hat \Phi(W)]^{n} W_x}\buildrel B_1 \over{\text{\textendash\textendash}}
\boxed{U_{t} = [\Phi(U)]^{n} U_{x}}\buildrel M \over{\text{\textendash\textendash}}\boxed{V_t = [\Psi(V)]^n V_{x}}\,.
\end{gather}
Note that (\ref{BCNH3}) is the natural non-Abelian counterpart of the corresponding part of the B\"acklund chart (\ref{BC1*}) introduced in~\cite{Fuchssteiner:Carillo:1989a}.

\section{On the non-Abelian mKdV equations}\label{s mkdvs}

A distinguished feature of the B\"acklund chart studied in this article is that it proceeds via two versions of the non-Abelian mKdV equation. The link between those versions is considered in this section. The alternative non-Abelian mKdV equation (amKdV)
\begin{gather}\label{mkdv2}
\widetilde V_t = \widetilde V_{xxx} + 3[\widetilde V,\widetilde V_{xx}] - 6 \widetilde V\widetilde V_x\widetilde V
\end{gather}
was f\/irst described in \cite{KK}, where the Lax pair formulation and the inverse scattering problem were studied. In contrast to the non-Abelian mKdV \eqref{mkdv} studied in the previous section, \eqref{mkdv2} does not admit a Miura transformation \cite{Olver:Sokolov}. Note also that, in the matrix case, \eqref{mkdv} is invariant under both $V\rightarrow V^*$ and $V\rightarrow -V$, whereas \eqref{mkdv2} is only invariant under $\widetilde V\rightarrow -\widetilde V^*$ ($V^*$ denoting the transpose of $V$).

As already pointed out in \cite{A}, these two versions of the mKdV equation are linked by the gauge transformation $\widetilde V=G^{-1}VG$ with $G_x=VG$. Consider the B\"acklund transformations% $B_2$ and $B_3$
\begin{gather}
B_2\colon \ G_x-VG =0 , \label{Bmkdvs1} \\
B_3\colon \ G_x-G\widetilde V =0 . \label{Bmkdvs2}
\end{gather}
Obviously, subsequent application of these B\"acklund transformations links \eqref{mkdv} to \eqref{mkdv2}.

The recursion operator $\widetilde{ \Psi}(\widetilde{V})$ of \eqref{mkdv2} is
\begin{gather} \label{psitilde}
\widetilde \Psi\big(\widetilde V\big) = \big(D+2C_{\widetilde V}\big) \big(D-2R_{\widetilde V}\big) \big(D+C_{\widetilde V}\big)^{-1} \big(D+2L_{\widetilde V}\big) D\big(D+C_{\widetilde V}\big)^{-1} .
\end{gather}
It was f\/irst given in \cite{GKS}, where it was derived using Lax representation. Here we use the B\"acklund link between the mKdV and the amKdV equation to give an alternative derivation together with a more conceptual formulation of the operator itself.

To this end, we introduce the derivation
\begin{gather} \label{der}
 \mathbb{D}:=D+C_{\widetilde V}.
\end{gather}

\begin{thrm} \label{t1} On use of the derivation \eqref{der}, the recursion operator \eqref{psitilde} of the amKdV equation can be written as
\begin{gather*}
\widetilde \Psi\big(\widetilde V\big) =
 \big(\mathbb{D}+C_{\widetilde V}\big)\big(\mathbb{D}-A_{\widetilde V}\big)\mathbb{D}^{-1}
 \big(\mathbb{D}+A_{\widetilde V}\big)\big(\mathbb{D}-C_{\widetilde V}\big)\mathbb{D}^{-1} .
\end{gather*}
\end{thrm}

\begin{rem} Recall that two operators $T$ and $S$ are called \emph{related} if they are of the form $T=AB$ and $S=BA$. In~\cite{Carillo:Schiebold:JMP2009}, relatedness of the non-Abelian KdV and mKdV recursion operators on the image of the Miura transform was exploited to derive solutions of the non-Abelian mKdV hierarchy from solutions of the non-Abelian KdV hierarchy.

Here we observe that $\Psi(V)$ and $\widetilde{\Psi}(\widetilde{V})$ are related in a generalized sense: If $P(D,V) = (D-C_V)D^{-1}$ and $Q(D,V) = (D+C_V)(D-A_V)D^{-1}(D+A_V)$, then
 \begin{gather*}
 \Psi(V) = P(D,V) Q(D,V), \qquad \widetilde{\Psi}\big(\widetilde{V}\big) = Q\big(\mathbb{D},\widetilde V\big) P\big(\mathbb{D},\widetilde V\big).
 \end{gather*}
\end{rem}

Before proving Theorem \ref{t1}, we observe some crucial properties of the derivation \eqref{der}.

\begin{prop} \label{p1} Let $V$ and $\widetilde V$ be related to each other via the B\"acklund transformations~$B_2$ and~$B_3$ with some intermediate invertible function $G$. Then the derivation~\eqref{der} is the result of conjugation of~$D$ with $K_G^{-1}$, i.e.,
\begin{gather*} \mathbb{D} = K_{G} D K_G^{-1} , \end{gather*}
where $K_G$ denotes the conjugation operator, $K_G T = G^{-1}TG$.

Moreover, it holds:
 \begin{itemize}\itemsep=0pt
 \item[$a)$] $K_G(D \pm C_V) K_G^{-1} = \mathbb{D} \pm C_{\widetilde V}$,
 \item[$b)$] $K_G (D \pm A_V) K_G^{-1} = \mathbb{D} \pm A_{\widetilde V}$.
 \end{itemize}
\end{prop}

\begin{proof} As an example, we verify $K_G D K_G^{-1} = \mathbb{D}$ and $K_G C_V K_G^{-1} = C_{\widetilde V}$.

Indeed,
\begin{gather*}
K_G D K_G^{-1} T = G^{-1} \big(GTG^{-1}\big)_x G = G^{-1} \big( G_xTG^{-1}+GT_xG^{-1} - GTG^{-1}G_xG^{-1} \big) G\\
 \hphantom{K_G D K_G^{-1} T}{}= G^{-1} G_xT + T_x - TG^{-1}G_x = \widetilde V T + T_x - T \widetilde V = \big(D + C_{\widetilde V}\big)T = \mathbb{D} T.
\end{gather*}

Analogously,
\begin{gather*}
K_G C_V K_G^{-1} T = G^{-1} \big[V,GTG^{-1}\big] G = G^{-1} \big( VGTG^{-1}-GTG^{-1}V \big) G \\
\hphantom{K_G C_V K_G^{-1} T}{} = G^{-1} VGT - TG^{-1}VG = \big[G^{-1}VG,T\big] = \big[\widetilde V, T\big] = C_{\widetilde V}T.\tag*{\qed}
\end{gather*}
\renewcommand{\qed}{}
\end{proof}

\begin{proof}[Proof of Theorem~\ref{t1}] On use of the B\"acklund transformations $B_2$, $B_3$, the amKdV recursion operator $\widetilde \Psi(\widetilde V)$ is related to the mKdV recursion operator $\Psi(V)$ via
\begin{gather} \label{rec with pi}
 \widetilde{\Psi}(\widetilde{V}) = \Pi \Psi(V) \Pi^{-1}
\end{gather}
with $\Pi=\Pi_3\Pi_2$, where
\begin{gather*}
 \Pi_2 = -(B_2)_{\vphantom{\widetilde V}G}^{-1}(B_2)_{\vphantom{\widetilde V}V}^{\vphantom{-1}}, \qquad \Pi_3= -(B_3)_{\widetilde V}^{-1}(B_3)_{\vphantom{\widetilde V}G}^{\vphantom{-1}}.
\end{gather*}
Let $L_V$ and $R_V$ denote left and right multiplication by $V$, then $C_V T=(L_V - R_V)T$ and $A_V T=(L_V + R_V)T$. Now it is straightforward to verify $(B_2)_G = (D-L_V)$, $(B_2)_V = -R_G$. This leads to the transformation operator $\Pi_2 = (D-L_V)^{-1} R_G$. Similarly, $(B_3)_G = (D-R_{\widetilde{V}})$, $(B_3)_{\widetilde{V}} = -L_G$ leading to $\Pi_3 = L_{G^{-1}} (D-R_{\widetilde{V}})$. Hence,
\begin{gather*}
 \Pi = L_{G^{-1}} \big(D-R_{\widetilde{V}}\big) (D-L_V)^{-1} R_G.
\end{gather*}
Next, using the identities $(D-L_V) R_G = R_G (D-C_V)$ and $(D-R_{\widetilde{V}}) R_G = R_G D$, which are straightforward to check, it follows
\begin{gather}
 \Pi = L_{G^{-1}} \big(D-R_{\widetilde{V}}\big) \big((D-L_V)^{-1} R_G \big) = L_{G^{-1}} \big( \big(D-R_{\widetilde{V}}\big) R_G \big) (D-C_V)^{-1}
 \nonumber \\
\hphantom{\Pi}{} = L_{G^{-1}} R_G D (D-C_V)^{-1} = K_{G} D (D-C_V)^{-1}. \label{pi}
\end{gather}
Hence, the amKdV recursion operator $\widetilde \Psi(\widetilde V)$ is obtained from the recursion operator $\Psi(V)$ of the mKdV equation as stated in~\eqref{mkdv-recop} via
\begin{gather*}
 \widetilde{ \Psi}\big(\widetilde{V}\big) \stackrel{\eqref{rec with pi}, \eqref{pi}}= \big( K_{G} D (D-C_V)^{-1} \big) (D-C_V)D^{-1}(D+C_V)(D-A_V)\\
 \hphantom{\widetilde{ \Psi}\big(\widetilde{V}\big)\stackrel{\eqref{rec with pi}, \eqref{pi}}=}{} \times D^{-1}(D+A_V) \big( (D-C_V) D^{-1} K_{G^{-1}} \big) \\
\hphantom{\widetilde{ \Psi}\big(\widetilde{V}\big)}{} = K_{G^{-1}} (D+C_V)(D-A_V)D^{-1}(D+A_V) (D-C_V) D^{-1} K_G \\
\hphantom{\widetilde{ \Psi}\big(\widetilde{V}\big)}{} = (\mathbb{D}+C_{\widetilde V})(\mathbb{D}-A_{\widetilde V})\mathbb{D}^{-1}
 (\mathbb{D}+A_{\widetilde V})(\mathbb{D}-C_{\widetilde V})\mathbb{D}^{-1} ,
\end{gather*}
where the last step requires Proposition~\ref{p1}.
\end{proof}

To conclude this section, for the reader's convenience we give a direct verif\/ication that the mKdV and the amKdV equation are related via $\widetilde{V} = G^{-1} G_x$ and $V=G_xG^{-1}$. More precisely, we show
\begin{itemize}\itemsep=0pt
 \item[a)] $V_t = \Psi(V)V_x \Longleftrightarrow \widetilde{V}_t = \widetilde \Psi(\widetilde{V})\widetilde{V}_x$,
 \item[b)] $ \widetilde \Psi(\widetilde{V})\widetilde{V}_x$ constitutes the right-hand side of \eqref{mkdv2}.
\end{itemize}
To prove a), recall that $\widetilde \Psi(\widetilde V) = \Pi \Psi(V) \Pi^{-1}$ where $\Pi = K_G D (D-C_V)^{-1} $, compare \eqref{rec with pi}, \eqref{pi} in the proof of Theorem~\ref{t1}. Noting that $\widetilde V = K_G V$ implies $\widetilde{V}_x = K_G V_x$, since $(G^{-1}G_x)_x = G^{-1} (G_xG^{-1})_x G$ (observe that $\widetilde{V}_{xx} = K_G V_{xx}$ does not hold). Hence
\begin{gather}
 \Pi^{-1} \widetilde{V}_x \stackrel{\eqref{pi}}{=} (D-C_V) D^{-1} K_G^{-1} \widetilde{V}_x = (D-C_V) D^{-1} V_x = (D-C_V) V = V_x. \label{h1}
\end{gather}
Furthermore, $\widetilde{V}_t = (G^{-1}G_x)_t = G^{-1}G_{xt} - G^{-1} G_t G^{-1} G_x = G^{-1} (G_{xt} G^{-1} - G_t G^{-1} G_x G^{-1}) G = K_G D(G_t G^{-1})$. This shows
\begin{gather}
 \Pi^{-1} \widetilde{V}_t = (D-C_V) D^{-1} K_G^{-1} \widetilde{V}_t = (D-C_V) \big(G_t G^{-1}\big) = V_t . \label{h2}
\end{gather}
Now a) follows from $ \widetilde \Psi(\widetilde V)\widetilde V_x = \Pi \Psi(V) \Pi^{-1} \widetilde V_x \stackrel{\eqref{h1}}{=} \Pi \Psi(V) V_x = \Pi V_t \stackrel{\eqref{h2}}{=} \widetilde V_t $.

To prove b), note f\/irst that $\widetilde V_x = D\widetilde V = (D+C_{\widetilde V}) \widetilde V = \mathbb{D} \widetilde V$. Similarly, $(\mathbb{D}-C_{\widetilde V})\widetilde V = \mathbb{D} \widetilde V$, and then $(\mathbb{D}+A_{\widetilde V})\mathbb{D} \widetilde V = \mathbb{D}^2 \widetilde V + \{\widetilde V, \mathbb{D}\widetilde V\} = \mathbb{D} (\widetilde V^2+\mathbb{D}\widetilde V)$. Using Theorem~\ref{t1}, we hence get
\begin{gather*}
 \widetilde{\Psi} \big(\widetilde{V}\big) \widetilde{V}_x
= \big(\mathbb{D}+C_{\widetilde V}\big) \big(\mathbb{D} - A_{\widetilde V}\big)\mathbb{D}^{-1}
 \big(\mathbb{D}+A_{\widetilde V}\big) \big(\mathbb{D}-C_{\widetilde V}\big) \mathbb{D}^{-1} \mathbb{D} \widetilde V \\
 \hphantom{\widetilde{\Psi} \big(\widetilde{V}\big) \widetilde{V}_x}{}
 = \big(\mathbb{D}+C_{\widetilde V}\big) \big(\mathbb{D} - A_{\widetilde V}\big)\mathbb{D}^{-1}
 \big(\mathbb{D}+A_{\widetilde V}\big) \mathbb{D} \widetilde V = \big(\mathbb{D}+C_{\widetilde V}\big) \big(\mathbb{D} - A_{\widetilde V}\big)
 \big( \mathbb{D}\widetilde V+\widetilde V^2 \big) \\
\hphantom{\widetilde{\Psi} \big(\widetilde{V}\big) \widetilde{V}_x}{}
= \big(\mathbb{D}+C_{\widetilde V}\big) \big( \mathbb{D}^2 \widetilde V - 2 \widetilde V^3 \big)
= \mathbb{D}^3 \widetilde V + \big[\widetilde V,\mathbb{D}^2 \widetilde V\big] - 2 \mathbb{D} \big(\widetilde V^3\big) .
\end{gather*}
Finally, evaluation of $\mathbb{D}=D+C_{\widetilde V}$ gives the right-hand side of the amKdV equation.

\section{A new non-Abelian B\"acklund chart}\label{s new}

In the last two sections a non-Abelian generalization of the B\"acklund chart linking pKdV, KdV and mKdV equations was presented, and extended to include the amKdV equation \eqref{mkdv2}. The aim of the present section is a further extension to also include non-Abelian analogues of the KdV singularity manifold equation and the KdV interacting soliton equation. To the best of the authors' knowledge, the non-Abelian equations presented in this sections are new. Together with the links presented in the previous section, the resulting non-Abelian B\"acklund chart\footnote{To facilitate the comparison with the Abelian case, a brief overview of the results in~\cite{Fuchssteiner:Carillo:1989a} is provided in the Appendix.} we obtain can be summarized as follows:
\begin{gather}
%\boxed{\mbox{pKdV}(W)} 	\ {\buildrel B_1 \over{\text{\textendash\textendash}}} \
%\boxed{\mbox{KdV} (U)} 	\ {\buildrel M \over{\text{\textendash\textendash}}} \
%\boxed{\mbox{mKdV}(V)} 	\ {\buildrel B_2, B_3 \over{\text{\textendash\textendash}}} \ \boxed{\mbox{amKdV}(\widetilde V)} \nonumber\\
%\qquad{}
%{\buildrel B_4\over{\text{\textendash\textendash}}} \ \boxed{\mbox{KdV sing.}(\phi)} \ {\buildrel B_5 \over{\text{\textendash\textendash}}} \
%\boxed{\mbox{int.\ soliton KdV}(S)}\, .\label{BC-nc}
\mbox{\scriptsize $\boxed{\mbox{pKdV}(W)} 	\, {\buildrel B_1\over{\text{\textendash\textendash}}} \,
\boxed{\mbox{KdV} (U)} 	\, {\buildrel M \over{\text{\textendash\textendash}}} \,
\boxed{\mbox{mKdV}(V)} 	\, {\buildrel B_2, B_3 \over{\text{\textendash\textendash}}} \, \boxed{\mbox{amKdV}(\widetilde V)}
{\buildrel B_4\over{\text{\textendash\textendash}}} \, \boxed{\mbox{KdV sing.}(\phi)} \, {\buildrel B_5 \over{\text{\textendash\textendash}}} \
\boxed{\mbox{int.\, soliton KdV}(S)}\, .$}\label{BC-nc}
\end{gather}
Moreover, the corresponding recursion operators are derived. Note that this chart extends to the respective hierarchies. In conclusion, it is shown that all recursion operators are hereditary. Finally, it is well-known that the KdV singularity manifold equation is invariant under the full M\"obius group, and it is shown that a generalized property holds for its non-Abelian counterpart.

To begin with, we introduce a non-Abelian analogue of the KdV singularity manifold equation (KdV sing.)
\begin{gather}\label{nc-KdVsing}
 \phi_t = \phi_x \{ \phi; x \} ,
\end{gather}
where $\{ \phi; x \}$ denotes the following non-Abelian version of the Schwarzian derivative
\begin{gather*}
 \{ \phi; x \} := \big( \phi_x^{-1}\phi_{xx} \big)_x - \frac{1}{2} \big( \phi_x^{-1}\phi_{xx} \big)^2.
\end{gather*}

\begin{rem}\quad
\begin{enumerate}\itemsep=0pt
\item[a)] To the best of the authors' knowledge, equation \eqref{nc-KdVsing} is new. For earlier non-Abelian ver\-sions of \eqref{nc-KdVsing} we refer to~\cite{Svinolupov}, see also~\cite{AF}. The dif\/ference to our approach is the non-Abelian interpretation of terms of the form $1/u_x^2$.
\item[b)] In the scalar case there is a very satisfying group theoretic explanation of the relation between the KdV sing.\ (also called UrKdV), mKdV and KdV equations~\cite{Wilson}, see also~\mbox{\cite{Depireux, Schiff}}. It seems to be an interesting problem to f\/ind an appropriate extension of these ideas to the non-Abelian setting. Departing from the present results, a f\/irst major dif\/f\/iculty is that our link between~\eqref{mkdv} and~\eqref{mkdv2} (which reduce to the same scalar equation) is not realized by an explicit mapping $V=F(\widetilde V)$.
\end{enumerate}
\end{rem}

The KdV sing.\ equation \eqref{nc-KdVsing} exhibits invariances analogous to its scalar counterpart. Note that the following theorem implies invariance under the full M\"obius group in the scalar case.

\begin{thrm} \label{t2} Let $A$, $B$, $C$, $D$ be constant operators. Then the non-Abelian KdV singularity manifold equation~\eqref{nc-KdVsing} is invariant under the transformations
\begin{gather*}
\phi \mapsto (\phi A+B)^{-1} (\phi C +D), \qquad \phi \mapsto (C\phi+D) (A\phi+B)^{-1} ,
 \end{gather*}
provided that $A$ and $D-BA^{-1}C$ are invertible.

In other words, if $\phi$ is an invertible solution of \eqref{nc-KdVsing}, then also $(\phi A+B)^{-1} (\phi C +D)$ and $(C\phi+D) (A\phi+B)^{-1} $ are solutions of~\eqref{nc-KdVsing}.
\end{thrm}

\begin{rem} The condition that $D-BA^{-1}C$ is invertible can be understood as a generalization of the scalar condition $ad-bc\not=0$ in the case $a\not=0$. Note however, that this is not the only possibility to guarantee the invariances of Theorem~\ref{t2}. For example, the assumption that $B$ and $D-AB^{-1}C$ are invertible is also suf\/f\/icient.
\end{rem}

The main ingredient of the proof of Theorem \ref{t2} is the following observation.

\begin{prop} \label{p2} The non-Abelian KdV singularity manifold equation \eqref{nc-KdVsing} is invariant under the transformation $\phi\mapsto\phi^{-1}$.
\end{prop}

\begin{proof} To start with we calculate the Schwarzian derivative of $\psi:=\phi^{-1}$.

Since $\psi_x =-\phi^{-1}\phi_x\phi^{-1}$, we get $\psi_{xx} = -\phi^{-1}\phi_{xx}\phi^{-1} +2\phi^{-1}\phi_x\phi^{-1}\phi_x\phi^{-1}$. Hence
\begin{gather*}
 \psi_x^{-1}\psi_{xx} = \phi\phi_x^{-1}\phi_{xx}\phi^{-1} - 2 \phi_x\phi^{-1},
\end{gather*}
and therefore
\begin{gather*}
 \{\psi; x\} = \big(\psi_x^{-1}\psi_{xx}\big)_x - \frac{1}{2} \big(\psi_x^{-1}\psi_{xx}\big)^2 \\
\hphantom{\{\psi; x\}}{} = \phi_{xx}\phi^{-1} + \boxed{ \phi \big(\phi_x^{-1}\phi_{xx} \big)_x \phi^{-1} }
 - \phi\phi_x^{-1}\phi_{xx}\phi^{-1}\phi_x\phi^{-1} - 2\phi_{xx}\phi^{-1} + 2\phi_x\phi^{-1}\phi_x\phi^{-1} \\
\hphantom{\{\psi; x\}=}{}
\boxed{ - \frac{1}{2} \phi \big(\phi_x^{-1}\phi_{xx}\big)^2 \phi^{-1} }
 + \phi \phi_x^{-1}\phi_{xx} \phi^{-1}\phi_x\phi^{-1} + \phi_{xx} \phi^{-1} - 2\phi_x\phi^{-1}\phi_x\phi^{-1}.
\end{gather*}
All terms with exception of the boxed ones cancel out, hence $\{\psi;x\} = \phi \{\phi, x\} \phi^{-1}$. As a result,
\begin{gather*}
 \psi_t = -\phi^{-1} \phi_t \phi^{-1} \stackrel{\eqref{nc-KdVsing}}{=} -\phi^{-1} \big( \phi_x \{ \phi;x \} \big) \phi^{-1}
 = \big({-} \phi^{-1}\phi_x \phi^{-1} \big) \big( \phi \{ \phi;x \} \phi^{-1} \big) = \psi_x \{ \psi; x\} ,
\end{gather*}
which completes the proof.
\end{proof}

The following invariance is straightforward to verify.

\begin{prop} \label{p3} The non-Abelian KdV singularity manifold equation \eqref{nc-KdVsing} is invariant under the transformation $\phi\mapsto A\phi B + C$ where $A$, $B$, $C$ are constant operators, provided that~$A$ and~$B$ are invertible.
\end{prop}

\begin{proof}[Proof of Theorem \ref{t2}] Let $\phi$ be a solution of~\eqref{nc-KdVsing}. By Proposition~\ref{p3}, also $\phi A+B$ solves~\eqref{nc-KdVsing} since~$A$ is invertible. Then, by Proposition~\ref{p2}, the inverse $\psi=(\phi A+B)^{-1}$ is a solution of~\eqref{nc-KdVsing}. Applying Proposition~\ref{p3}, $\psi(D-BA^{-1}C)+A^{-1}C = (\phi A+B)^{-1}(D-BA^{-1}C+(\phi A+B)A^{-1}C) = (\phi A+B)^{-1}(\phi C+D)$ is a solution of~\eqref{nc-KdVsing}.
\end{proof}

Finally, we introduce the non-Abelian KdV interacting soliton equation (int.\ soliton KdV)
 \begin{gather}\label{intS}
S_t = \left( S \left( \big(S^{-1}S_x\big)_x - \frac{1}{2} \big(S^{-1}S_x\big)^2 \right) \right)_x .
\end{gather}
Again, to the best of the authors' knowledge, this equation is new.

Analogously as in Section~\ref{s mkdvs}, it can be verif\/ied that the amKdV \eqref{mkdv2}, KdV sing.~\eqref{nc-KdVsing} and int.\ soliton KdV~\eqref{intS} equations are related by the B\"acklund transformations
\begin{gather}
B_4\colon \ 2 \phi_x {\widetilde V} - \phi_{xx} =0 , \label{B-tildeVphi} \\
B_5\colon \ \phi_x-S =0 . \label{BSphi}
\end{gather}

\begin{rem} Note that the counterpart of the transformation $\widetilde{V} = \frac{1}{2}\phi_x ^{-1}\phi_{xx}$ from \eqref{B-tildeVphi}, namely, $\widehat{V} = -\frac{1}{2}\phi_{xx}\phi_x ^{-1}$, also maps solutions $\phi$ of the KdV sing.\ equation~\eqref{nc-KdVsing} to solutions $\widehat{V}$ of the amKdV equation~\eqref{mkdv2}.
\end{rem}

Both the recursion operators $\Upsilon(\phi)$ of \eqref{nc-KdVsing} as well as $\chi(S)$ of \eqref{intS} can be constructed on use of the B\"acklund links (\ref{B-tildeVphi}), (\ref{BSphi}).

\begin{thrm} \label{t rec ops} Define $N(T) = \frac{1}{2} T^{-1}T_x$. Then the recursion operators $\Upsilon (\phi)$ of KdV sing.~\eqref{nc-KdVsing} and $\chi(S)$ of~\eqref{intS} are given by
 \begin{gather}
 \Upsilon (\phi) = L_{\phi_x} \mathbb{D}^{-1} \big(\mathbb{D}-A_{N(\phi_x)}\big) \big(\mathbb{D}-C_{N(\phi_x)}\mathbb{D}^{-1} C_{N(\phi_x)}\big)
 \big(\mathbb{D}+A_{N(\phi_x)}\big) L_{\phi_x^{-1}} , \label{rec sing} \\
 \label{rec int} \chi(S)= \textstyle L_{S} \big(\mathbb{D}-A_{N(S)}\mathbb{D}^{-1}A_{N(S)}\big)
 \big(\mathbb{D}-C_{N(S)}\mathbb{D}^{-1} C_{N(S)}\big) L_{S}^{-1} .
 \end{gather}
 Note that $\mathbb{D} = D+C_{N(\phi_x)} = D+C_{N(S)} $.
\end{thrm}

\begin{proof}Indeed, from the B\"acklund links we have that the recursion operators satisfy
\begin{alignat*}{3}
 &\Upsilon (\phi) =\Pi_4 \widetilde{\Psi} \big(\widetilde{V}\big)\Pi_4^{-1}, \qquad && \Pi_4 = - (B_4)_{\phi}^{-1} (B_4)_{\widetilde{V}}, & \\
 & \chi(S)= \Pi_5 \Upsilon (\phi) \Pi_5^{-1}, \qquad && \Pi_5 = - (B_5)_S^{-1} (B_5)_{\phi} .&
\end{alignat*}
It is straightforward to calculate $(B_4)_{\widetilde V} = 2L_{\phi_x}$, $(B_4)_{\phi} = -(D-2R_{\widetilde V})D$. Therefore, we get
\begin{gather*}
\Pi_4 = 2 D^{-1}\big(D-2R_{\widetilde{V}}\big)^{-1} L_{\phi_x} .
\end{gather*} Similarly $(B_5)_{S} = -I$, $(B_5)_{\phi} = D$, and hence $\Pi_5=D$. Consequently,
\begin{gather*}
 \Pi_5 \Pi_4 = 2 \big(D-2R_{\widetilde{V}}\big)^{-1} L_{\phi_x} = 2 \big(\mathbb{D}-A_{\widetilde{V}}\big)^{-1} L_{S} .
\end{gather*}
Using the identity $(\mathbb{D}-A_{\widetilde V})L_S = L_S(\mathbb{D}+C_{\widetilde V})$, which can be checked directly using only the product rule and $\widetilde V = \frac{1}{2}S^{-1}S_x$, we f\/ind
\begin{gather*}
 \Pi_5 \Pi_4 = 2 L_S \big(\mathbb{D} + C_{\widetilde{V}}\big)^{-1} .
\end{gather*}
From this and Theorem \ref{t1}, we get
\begin{gather*}
 \chi (S) = (\Pi_5 \Pi_4) \widetilde{\Psi}\big(\widetilde{V}\big) (\Pi_5 \Pi_4 )^{-1}
 = L_S \big(\mathbb{D}+C_{\widetilde{V}}\big)^{-1} \widetilde{\Psi}\big(\widetilde{V}\big) \big(\mathbb{D}+C_{\widetilde{V}}\big)L_S^{-1} \\
\hphantom{\chi (S)}{}
 = L_S \big(\mathbb{D}-A_{\widetilde{V}}\big) \mathbb{D}^{-1} \big(\mathbb{D}+A_{\widetilde{V}}\big)
 \big(\mathbb{D}-C_{\widetilde{V}}\big)\mathbb{D}^{-1} \big(\mathbb{D}+C_{\widetilde{V}}\big) L_{S^{-1}} .
\end{gather*}
The claim for $\chi(S)$ follows from observing that $(\mathbb{D}-A_{\widetilde{V}}) \mathbb{D}^{-1} (\mathbb{D}+A_{\widetilde{V}}) = \mathbb{D} -A_{\widetilde{V}}\mathbb{D}^{-1}A_{\widetilde{V}} = \mathbb{D} - A_{N(S)} \mathbb{D}^{-1} A_{N(S)} $ and $(\mathbb{D}-C_{\widetilde{V}}) \mathbb{D}^{-1} (\mathbb{D}+C_{\widetilde{V}}) = \mathbb{D} -C_{\widetilde{V}}\mathbb{D}^{-1}C_{\widetilde{V}} = \mathbb{D} - C_{N(S)} \mathbb{D}^{-1} C_{N(S)} $. Finally,
\begin{gather*}
 \Upsilon (\phi) = \Pi_5^{-1} \chi(S) \Pi_5 = D^{-1} L_S \big( \mathbb{D} - A_{N(S)} \mathbb{D}^{-1} A_{N(S)} \big)\big(\mathbb{D} -
 C_{N(S)} \mathbb{D}^{-1} C_{N(S)}\big) L_S^{-1} D.
\end{gather*}
Since $DL_{S}^{-1} = L_S^{-1}D - L_S^{-1}L_{S_x} L_S^{-1}$, the identity $L_S^{-1}D = (D+L_{S^{-1}S_x})L_S^{-1} = (D+2L_{N(S)})L_S^{-1} = (\mathbb{D}+A_{N(S)})L_S^{-1} $ holds, and thus
\begin{gather*}
 \Upsilon (\phi) = L_S \big(\mathbb{D}+A_{N(S)}\big)^{-1} \big( \mathbb{D} - A_{N(S)} \mathbb{D}^{-1} A_{N(S)} \big)\big(\mathbb{D} -
 C_{N(S)} \mathbb{D}^{-1} C_{N(S)}\big) \big(\mathbb{D}+A_{N(S)}\big) L_S^{-1} \\
\hphantom{\Upsilon (\phi)}{}
= L_{S} \mathbb{D}^{-1} \big(\mathbb{D}- A_{N(S)}\big) \big(\mathbb{D}-C_{N(S)}\mathbb{D}^{-1} C_{N(S)}\big) \big(\mathbb{D}+A_{N(S)}\big) L_S^{-1},
\end{gather*}
which is the claim for $\Upsilon (\phi) $ upon substituting $S=\phi_x$.
\end{proof}

In the scalar case, the recursion operator for the KdV sing.\ equation \eqref{nc-KdVsing} is well-known, see, e.g.,~\cite{Fuchssteiner:Carillo:1989a}, where it appears in the form
 \begin{gather*}
 \Upsilon(\phi) = \phi_x D^{-1} \big(D^3 + 2D{\sf U}+2{\sf U}D \big) \phi_x^{-1}.
 \end{gather*}
with ${\sf U}=\frac{1}{2} (\phi_{xx}/\phi_x)_x - \frac{1}{4} (\phi_{xx}/\phi_x)^2$. To facilitate comparison to the non-Abelian case, we state the following reformulation of \eqref{rec sing}.

\begin{cor} The recursion operator $\Upsilon (\phi)$ of the KdV sing.\ equation \eqref{nc-KdVsing} in Theorem~{\rm \ref{t rec ops}} can be written as
 \begin{gather*}
 \Upsilon (\phi) = L_{\phi_x} \mathbb{D}^{-1} \big( \mathbb{D}^3 + \mathbb{D} A_{\sf U} + A_{\sf U} \mathbb{D} + C_{\sf U} \mathbb{D}^{-1} C_{\sf U}
 \big) L_{\phi_x^{-1}} ,
 \end{gather*}
 where ${\sf U} = ( N(\phi_x) )_x - ( N(\phi_x) )^2$.
\end{cor}

\begin{proof} The corollary is an immediate consequence of Proposition \ref{factor}, the fact that $N(\phi_x) = \widetilde{V}$ and $\widetilde{V}_x= \mathbb{D}\widetilde{V}$.
\end{proof}

\begin{rem} Note that in the scalar case ${\sf U} = v_x -v^2$, i.e., $\sf U$ stems from a Miura transformation of a solution~$v$ to the mKdV equation. In the non-Abelian setting, we have ${\sf U} = \widetilde{V}_x -\widetilde{V}^2$, which means that~$\sf U$ stems from a (non-Abelian) Miura transformation of a~solution $\widetilde{V}$ to the amKdV equation~\eqref{mkdv2}.
\end{rem}

\begin{prop} \label{factor} Let $\sf D$ be a derivation, and $\sf U$, $\sf V$ related by the Miura transform ${\sf U} = {\sf D}{\sf V}-{\sf V}^2$. Then the following factorization holds:
 \begin{gather*}
 {\sf D}^3 + A_{\sf U} {\sf D} + {\sf D} A_{\sf U} + C_{\sf U} {\sf D}^{-1} C_{\sf U}
 = ({\sf D}-A_{\sf V})\big({\sf D}-C_{\sf V}{\sf D}^{-1} C_{\sf V}\big)({\sf D}+A_{\sf V}) \\
\hphantom{{\sf D}^3 + A_{\sf U} {\sf D} + {\sf D} A_{\sf U} + C_{\sf U} {\sf D}^{-1} C_{\sf U}}{}
= ({\sf D}-A_{\sf V}) ({\sf D} \pm C_{\sf V}) {\sf D}^{-1} ({\sf D} \mp C_{\sf V}) ({\sf D}+A_{\sf V}).
 \end{gather*}
\end{prop}

\begin{proof}The proof of the proposition is completely analogous to the proof of the factorization of the non-Abelian KdV recursion operator on the image of a~Miura transformation in \cite[Proposition~14]{Carillo:Schiebold:JMP2009}. The crucial observation is that due to the product rule the identities
\begin{gather*}
 {\sf D} A_T = A_T {{\sf D}} + A_{{\sf D}T} , \qquad {\sf D} C_T = C_T {{\sf D}} + C_{{\sf D}T}
\end{gather*}
hold for any derivation ${\sf D}$.
\end{proof}

{\sloppy There are close algebraic relations between the non-Abelian recursion operators of the amKdV, KdV sing.\ and int.\ soliton KdV equations.

}

\begin{cor} For $\widetilde V =\frac{1}{2} \phi_x^{-1} \phi_{xx}$ and $S=\phi_x$, the recursion operators in Theorem~{\rm \ref{t1}} and the recursion operators in~\eqref{rec sing},~\eqref{rec int} are pairwise related. More precisely, it holds
 \begin{gather*}
 \widetilde \Psi\big(\widetilde V\big) = P_1P_2= Q_1 Q_2, \qquad \chi(S) = Q_2 Q_1 = R_1 R_2, \qquad \Upsilon(\phi) = P_2P_1 = R_2 R_1,
 \end{gather*}
 where
 \begin{alignat*}{3}
& P_1 = \big(\mathbb{D}+C_{\widetilde V}\big) \big(\mathbb{D}+A_{\widetilde V}\big) L_{\phi_x}^{-1}, \qquad &&
 		 P_2 = L_{\phi_x} \mathbb{D}^{-1} \big(\mathbb{D}-A_{\widetilde V}\big) (\mathbb{D}- C_{\widetilde V}) \mathbb{D}^{-1}, & \\	
&	Q_1 = \big(\mathbb{D}+C_{\widetilde V}\big) L_{\phi_x}^{-1} , \qquad && Q_2 = L_{\phi_x} \big(\mathbb{D}-A_{\widetilde V}\mathbb{D}^{-1} A_{\widetilde V}\big)\big(\mathbb{D}-C_{\widetilde V}\big) \mathbb{D}^{-1} ,& \\
&	R_1 = L_{\phi_x} \big(\mathbb{D}+A_{\widetilde V}\big) L_{\phi_x}^{-1}, \qquad && R_2 = L_{\phi_x} \mathbb{D}^{-1} \big(\mathbb{D}-A_{\widetilde V}\big) \big(\mathbb{D} -	 C_{\widetilde V}\mathbb{D}^{-1}C_{\widetilde V}\big)L_{\phi_x}^{-1}.&
 \end{alignat*}
\end{cor}

\section{Hereditariness and hierarchies}\label{section5}

The links obtained in the last two sections are summarized in
\begin{gather*} %\label{BC-nc part 2}
\boxed{\mbox{mKdV}(V) } \ {\buildrel B_2, B_3 \over{\text{\textendash\textendash}}} \
\boxed{\mbox{amKdV}(\widetilde V)} \ {\buildrel B_4\over{\text{\textendash\textendash}}} \
\boxed{\mbox{KdV sing.}(\phi)}\ {\buildrel B_5 \over{\text{\textendash\textendash}}} \
\boxed{\mbox{int.\ soliton KdV}(S)}\,,
\end{gather*}
where the B\"acklund transformations $B_2$, $B_3$, $B_4$ and $B_5$ are given in \eqref{Bmkdvs1}, \eqref{Bmkdvs2}, \eqref{B-tildeVphi} and \eqref{BSphi}, respectively. Applying the respective recursion operators, the B\"acklund chart extends to the whole hierarchies as follows
\begin{gather} \label{BC-nc part 2 H}
\boxed{V_t = [\Psi(V)]^n V_{x}}		\buildrel B_2, B_3 \over{\text{\textendash\textendash}}
\boxed{\widetilde V_t = [\widetilde \Psi(\widetilde V)]^n \widetilde V_{x}}	
\buildrel B_4 \over{\text{\textendash\textendash}}
\boxed{\phi_t = [\Upsilon(\phi)]^n \phi_{x}}	\buildrel B_5 \over{\text{\textendash\textendash}}
\boxed{S_t = [\chi(S)]^n S_x}\, .
\end{gather}

As known from the scalar case, hereditariness is a crucial property of recursion operators \cite{Fuchssteiner1979, Magri}. Unfortunately its direct verif\/ication often requires involved computations, in particular in the non-Abelian case. The following proposition uses the B\"acklund links established in the present article to avoid computations by reducing the proof of hereditariness to the hereditariness of the non-Abelian KdV recursion operator, which is proved in~\cite{Schiebold2010}.

\begin{prop} \label{p4} Each of the recursion operators $\Psi$ from the B\"acklund chart \eqref{BC-nc} has the following properties:
\begin{enumerate}\itemsep=0pt
 \item[$a)$] $\Psi$ is hereditary,
 \item[$b)$] $\Psi$ is a strong symmetry for all equations of the hierarchy generated by $\Psi$.
\end{enumerate}
\end{prop}

\begin{proof}
In \cite[(32)]{Schiebold2010} it is verif\/ied that the non-Abelian KdV recursion operator \eqref{kdv-recop} satisf\/ies the identity
\begin{gather*} [D, \Phi (U)] = \Phi'(U)[U_x] , \end{gather*}
implying that $\Phi(U)$ is a strong symmetry (recursion operator in the sense of~\cite{Olver}) for the trivial member $U_{t}=U_x$ of the non-Abelian KdV hierarchy \cite{Fuchssteiner1979}. Moreover, the main result in \cite{Schiebold2010} is that $\Phi(U)$ is hereditary. Hence $\Phi(U)$ is a strong symmetry for all equations of the non-Abelian KdV hierarchy \cite{Fuchssteiner1979}. As shown in \cite{FoFu2}, the properties of a) and b) are preserved under B\"acklund transformations.
\end{proof}

Each of the hierarchies in \eqref{BC-nc part 2 H} is of the form
\begin{gather} \label{hier}
 U_t = [\Phi(U)]^n U_x .
\end{gather}
We may rewrite the right-hand side of \eqref{hier} as
\begin{gather*} [\Phi(U)]^n U_x = X_n(U) ,\end{gather*}
where $X_n$ is a vector f\/ield on the space of $x$-dependent operator-valued functions (see \cite{Schiebold2010} for details). The main consequence of Proposition~\ref{p4} is

\begin{cor} Let $X_1, X_2,\ldots $ be the vector fields of one of the hierarchies in \eqref{hier}. Then we have
 \begin{gather*} [X_m,X_n] = 0 \end{gather*}
for all $m,n =1,2,\ldots$.
\end{cor}

The proof uses arguments explained in \cite[Section VI]{Schiebold2010}.

\section{Remarks, perspectives and open problems} \label{rems}
This section collects some remarks on the results previously presented together with some perspectives study and open problems. The chain of B\"acklund transformations we obtained, represents a not at all trivial generalization to the operator level of the corresponding one~\cite{Fuchssteiner:Carillo:1989a} which links the scalar pKdV, KdV, mKdV, KdV interacting soliton and KdV singularity mani\-fold hierarchies. Furthermore, it generalizes the non-Abelian B\"acklund chart in \cite{Carillo:Schiebold:JMP2009} since it connects further non-Abelian hierarchies. Specif\/ically, new non-Abelian equations, and, then, the corresponding hierarchies, arise, such as the amKdV, in~\eqref{mkdv2}.

\bigskip\noindent{\bf Remarks}
\begin{itemize}\itemsep=0pt
\item The two B\"acklund charts, respectively, in the Abelian \cite{Fuchssteiner:Carillo:1989a} (see Appendix) and the non Abelian case, connect the pKdV, KdV, mKdV, KdV-singularity manifold and Interacting Soliton KdV equations or their non-Abelian analogs.
\item Known the recursion operator of a nonlinear evolution equation then all the other nonlinear evolution equations linked to it via a B\"acklund chart admit a recursion operator. The latter can be constructed in the Abelian as well as in the non-Abelian case.
\item All nonlinear evolution equations in the same B\"acklund chart share all the structural pro\-perties which are preserved under B\"acklund transformations as soon as a~single equation, in it, enjoys them. Remarkable is the hereditariness of the recursion operators~\cite{Schiebold2010}.
\item Also in the non-Abelian operator case, given the hereditary recursion operators, the B\"acklund chart can be extended to the corresponding generated hierarchies. Again, the B\"acklund chart relates the corresponding members of each one of the involved hierarchies of nonlinear evolution equations.
\item Even if there are similarities between the Abelian scalar case and the non-Abelian operator case, in the second case the structure is richer. Indeed, two distinct non-Abelian mKdV equations appear: the non-Abelian mKdV and an {\it alternative} mKdV equations which do coincide in the Abelian case. Correspondingly, when commutativity is assumed, combination of the B\"acklund transformations $B_2$ and $B_3$ produces the identity transformation and, hence, the Abelian B\"acklund chart~\cite{Fuchssteiner:Carillo:1989a} is recovered.
\item A similar behavior can be observed also when the Cole--Hopf link connecting Burgers equation to linear heat equation is extended to the non-Abelian case \cite{SIMAI2008, Carillo:Schiebold:JNMP2012}. The heat equation is connected to two dif\/ferent Burgers equations, termed Burgers and mirror Burgers equations in \cite{CMS-2016a, Ku}. Then, recursion operators and the corresponding hierarchies follow from the Cole--Hopf link, which can be regarded as a special case of B\"acklund transformation.
\end{itemize}

\noindent{\bf Perspectives and open problems}
\begin{itemize}\itemsep=0pt
\item We expect a similar situation when the 5th order nonlinear evolution equations which appear in the B\"acklund chart in \cite{BS1, Rogers:Carillo:1987b} are extended to the non-Abelian case. This study is currently under investigation and we are devising also computer aided routines to check the algebraic properties of the recursion operators. Indeed, already in the Abelian case, the computations involved are very long and complicated.
\item Furthermore, if a nonlinear evolution equation admits a Hamiltonian and bi-Hamiltonian structure, related to the recursion operator \cite{[12], {Fuchssteiner:Carillo:1990a}, Benno-Walter, Magri}, then all nonlinear evolution equations in the same B\"acklund chart admit a Hamiltonian and bi-Hamiltonian structure.
\item The approach can be extended also when nonlinear evolution equations in $(2+1)$ dimensions are considered, namely, the unknown function is supposed to depend on two space variables rather than on a single space variable. Thus, in \cite{walsan1}, the Kadomtsev--Petviashvili (KP), modif\/ied Kadomtsev--Petviashvili (mKP) and $(2+1)$-dimensional Harry Dym equations, which represent, in turn, the $(2+1)$-dimensional analog of KdV, mKdV, and Harry Dym equations are all connected via B\"acklund transformations. Notably, the connection among their $(1+1)$-dimensional corresponding equations \cite{walsan2} follows on imposing suitable constraints to the $(2+1)$-dimensional ones.
\end{itemize}

\appendix
\section{Appendix} \label{ex-sec2}

In this Appendix some background def\/initions which are of use throughout the whole article are brief\/ly recalled in the opening subsection: this choice is due to the lack of uniqueness in many def\/initions in the literature. In the following subsection, some results strictly connected to the present investigation are retrieved.

\subsection{Background notions}

Here, to improve readability, some fundamental def\/initions are provided. First of all, we recall the notion of B\"acklund transformation, according to the def\/inition of Fokas and Fuchssteiner~\cite{FokasFuchssteiner:1981} (see also the book by Rogers and Shadwick~\cite{RogersShadwick}).
Consider non linear evolution equations of the type
\begin{gather}\label{1}
u_t = K ( u ),
\end{gather}
where the unknown function $u$ depends on the independent variables~$x$,~$t$ and, for f\/ixed~$t$, \mbox{$u (x,t) \in M$}, a~manifold modeled on a linear topological space so that the generic {\it fiber} $T_uM$, at $u\in M$, can be identif\/ied with $M$ itself\footnote{It is generally assumed that~$M$ is the space of functions $u(x,t)$ which, for each f\/ixed~$t$, belong to the Schwartz space $S$ of {\it rapidly decreasing functions} on ${\mathbb R}^n$, i.e., $S({\mathbb R}^n):=\{ f\in C^\infty({\mathbb R}^n)\colon \| f \|_{\alpha,\beta} < \infty, \forall \alpha,\beta\}$, where $\| f \|_{\alpha,\beta}:= \sup\limits_{x\in{\mathbb R}^n} \left\vert x^\alpha D^\beta f(x) \right\vert $, and $D^\beta:=\partial^\beta /{\partial x}^\beta$. In the present section $n=1$.}, and $ K \colon M \rightarrow TM$, is an appropriate $ C^{\infty}$ vector f\/ield on a~mani\-fold $ M$ to its tangent manifold~$TM$.

Let, now
\begin{gather}\label{2}
v_t = G (v)
\end{gather}
denote a second nonlinear evolution equation. Accordingly, in turn, when $ u (x,t) \in M_1$, $v (x,t) \in M_2 $ and $M_1$, $M_2$ represent manifolds modeled on a~linear topological space then, $ K \colon M_1 \rightarrow TM_1 $ and $ G\colon M_2 \rightarrow TM_2 $ denote appropriate $ C^{\infty}$-vector f\/ields on the manifolds $M_i$, $i=1,2$,
\begin{alignat*}{4}%\label{eq.s}
& u_t = K ( u ),\qquad && K \colon \ M_1 \rightarrow TM_1,\qquad && u \colon \ (x,t) \in {\mathbb R} \times {\mathbb R}\to u (x,t) \in M_1, & \\
& v_t = G (v), \qquad && G \colon \ M_2 \rightarrow TM_2, \qquad && v \colon \ (x,t) \in {\mathbb R} \times {\mathbb R} \to v (x,t) \in M_2 .&
\end{alignat*}
As usual, when soliton solutions are considered, the further assumption $M:= M_1\equiv M_2$ is adopted. Then, the def\/inition of B\"acklund transformation, according to~\cite{FokasFuchssteiner:1981}, can be stated as follows.

\begin{Definition} Given two evolution equations, $ u_t = K (u)$ and $v_t = G (v)$, then $B (u , v) = 0$ represents a B\"acklund transformation between them if, whenever two solutions of these equations are given, let denote them $u(x,t)$ and $v(x,t)$ respectively, such that $B (u, v) = 0$ at the initial time $t=0$, then this holds true for all times, namely,
\begin{gather*}
B (u(x,t), v(x,t)) \vert_{ t=0 } = 0 \quad \Longrightarrow \quad
B (u(x,t),v(x,t) )\vert_{t=\bar t} = 0, \qquad \forall\, \bar t >0, \quad \forall \, x\in{\mathbb R}.
\end{gather*}
\end{Definition}

Therefore, the B\"acklund transformation $B$ establishes a~correspondence between solutions of the evolution equations it connects. This link is graphically represented as
\begin{gather*}%\label{BC1}
\boxed{u_t = K (u)} \ {\buildrel B \over {\text{\textendash\textendash}}} \ \boxed{ v_t = G (v) }
\end{gather*}
and, then, if the nonlinear evolution equation $u_t = K (u)$ admits a hereditary recursion opera\-tor~$\Phi (u)$ \cite{FokasFuchssteiner:1981}, namely,
\begin{gather*}%\label{base_u}
u_t = \Phi( u ) u_x, \qquad \text{where} \quad K (u) = \Phi (u) u_x,
\end{gather*}
also the equation $v_t=G(v)$ admits a hereditary recursion operator, say~$\Psi(v)$. The B\"acklund transformation $B$ allows to f\/ind $\Psi(v)$ via
\begin{gather}\label{transf-op}
\Psi(v)= \Pi \Phi (u) \Pi^{ -1}, \qquad \text{where} \quad \Pi : = -B_v^{ -1} B_u, \qquad \Pi \colon \ T M_1 \rightarrow T M_2,
\end{gather}
where $B_u$ and $B_v$ denote the Frechet derivatives of the B\"acklund transformation $B(u,v)$. The B\"acklund transformation~$B$~\cite{FokasFuchssteiner:1981}, via relation~(\ref{transf-op}), guarantees that $\Psi(v)$ represents the here\-ditary recursion operator admitted by the equation $v_t=G(v)$. Hence, the two hierarchies $u_{t} = [\Phi (u) ]^{n} u_x$ and $v_{t }= [\Psi (v)]^{n} v_x$, $n\geq 0$, of nonlinear evolution equations can be constructed~ \cite{Fuchssteiner1979} and their {\it base members} equations are, in turn, (\ref{1}) and~(\ref{2}). Each equation in these two hierarchies, parametrized\footnote{Here and throught ${\mathbb N}_0:= {\mathbb N}\cup\{0\}$.} by $n\in{\mathbb N}_0$, is connected, via the same B\"acklund transformation $B$, to the corresponding equation, parametrized by the same~$n$, in the other hierarchy. This extension to the whole hierarchies is depicted in the following B\"acklund chart
\begin{gather*}%\label{BT-hier}
\boxed{u_{} = [\Phi (u) ]^{n} u_x } \ {\buildrel B \over {\text{\textendash\textendash}}} \ \boxed{ v_{t} = [\Psi (v) ]^{n} v_x}\, ,
\end{gather*}
which emphasizes that the link between the two equations (\ref{1}) and (\ref{2}) is inherited\footnote{These recursion operators are termed {\it hereditary} in~\cite{Fuchssteiner1979} to stress that the property to be recursion operator, established referring to the base member equation, is inherited by all the other members in the hierarchy of nonlinear evolution equations.} by each member of the two hierarchies generated, respectively, by the recursion operators~$\Phi$ and~$\Psi$.

\subsection{Abelian KdV connected B\"acklund chart}
Here some results, obtained in \cite{Fuchssteiner:Carillo:1989a}, concerning the links among nonlinear evolution equations whose base member is a 3-rd order equation, are recalled. An overview on further previous result is given in~\cite{CarilloActaAM2012}. Thus, the B\"acklund chart here below, (\ref{BC1*}), allows to f\/ind solutions of Harry Dym equation~\cite{BS4}. This method to f\/ind solutions is not restricted to third order nonlinear evolution equations; indeed, for instance, also solutions of Burgers equation are obtained via B\"acklund transformations (see \cite{Cole:1951, Guo:Carillo, Hopf:1950, LeviRB}) to mention a few of them.

The following B\"acklund chart summarizes links relating 3rd order nonlinear evolution equations, when the links in \cite{Fuchssteiner:Carillo:1989a} are combined with the link between the KdV and potential KdV (pKdV),
\begin{gather}
%\boxed{\text{pKdV}(q)}\ {\buildrel (a) \over{\text{\textendash\textendash}} } \ \boxed{\text{KdV}(u)}\ {\buildrel (b)
%\over{\text{\textendash\textendash}}} \ \boxed{\text{mKdV}(v) } \ {\buildrel (c) \over{\text{\textendash\textendash}}} \ \boxed{\text{KdV~sing.}(\varphi)}
%\nonumber\\
%\qquad{} {\buildrel (d) \over{\text{\textendash\textendash}}}\ \boxed{\text{int. soliton KdV}(s) } \
%{\buildrel (e) \over{\text{\textendash\textendash}}} \ \boxed{\text{Dym}(\rho)}\,,\label{BC1*}
\mbox{\small $\boxed{\!\text{pKdV}(q)\!}\, {\buildrel (a) \over{\text{\textendash\textendash}} } \, \boxed{\!\text{KdV}(u)\!}\, {\buildrel (b)
\over{\text{\textendash\textendash}}} \, \boxed{\!\text{mKdV}(v)\! } \, {\buildrel (c) \over{\text{\textendash\textendash}}} \, \boxed{\!\text{KdV~sing.}(\varphi)\!}
 {\buildrel (d) \over{\text{\textendash\textendash}}}\, \boxed{\!\text{int. soliton KdV}(s) \!} \,
{\buildrel (e) \over{\text{\textendash\textendash}}} \, \boxed{\!\text{Dym}(\rho)\!}\,,$}\label{BC1*}
\end{gather}
wherein the third order nonlinear evolution equations are, in turn, given by
\begin{alignat*}{3}
& q_t = q_{xxx} + 3 q^2_x \qquad && \text{(pKdV)}, & \\
& u_t = u_{xxx} + 6 uu_x \qquad && \text{(KdV)}, & \\
& v_t = v_{xxx} - 6 v^2 v_x \qquad && \text{(mKdV)},& \\
& \varphi_t = \varphi_x \{ \varphi ; x\} , \quad \text{where} \ \ \{ \varphi ; x \} :=
 \left( { \varphi_{xx} \over \varphi_x} \right)_x -
{1 \over 2 }\left({ \varphi_{xx} \over \varphi_x} \right)^2 \qquad && \text{(KdV~sing.)}, & \\
& s^2 s_t = s^2 s_{xxx} - 3 s s_x s_{xx}+ {3 \over 2 }{s_x}^3 \qquad && \text{(int.\ soliton~KdV)}, & \\
& \rho_t = \rho^{3} \rho_{\xi \xi \xi} \qquad && \text{(Dym)},&
\end{alignat*}
in addition, the B\"acklund transformations which link them, are, respectively:
 \begin{gather*}
(a) \ \ u-q_x = 0 ,\qquad (b) \ \ u + v_x + v^2 =0 , \qquad (c) \ \ v -{ {1 \over 2 }{ \varphi_{xx} \over \varphi_x} } = 0,\\
(d) \ \ s - \varphi_x =0, \qquad (e) \ \ {\bar x} : = D^{-1} s (x) \qquad \text{where} \quad D^{-1}:=
\int_{-\infty}^x d\xi .
\end{gather*}
The transformation $(e)$ is a {\it reciprocal\/}\footnote{Reciprocal-type transformations were originally introduced in $(2+0)$-dimensional isentropic gasdynamics by Bateman \cite{Bateman1, Bateman2} and in $(1+1)$-dimensional anisentropic gasdynamics by Rogers~\cite{Rogers}. Connection with invariance of these gasdynamic systems under B\"acklund transformations is described in \cite{RogersShadwick}.} transformation \cite{Bateman2, RogersShadwick}: according to \cite{BS1,Fuchssteiner:Carillo:1989a}, it represents a B\"acklund transformation whose transformation operator~$\Pi$ can be constructed, provided it is considered an extended manifold which comprises also the independent space variable $x$, further to the dependent variable~$u$. That is, while the B\"acklund transformations $(a)$, $(b)$ and $(c)$ admit operators~$\Pi$
\begin{gather*}
\Pi_a \colon \ T M_u \rightarrow T M_v, \qquad \Pi_b \colon \ T M_v \rightarrow T M_\varphi, \qquad \Pi_c\colon \ T M_\varphi \rightarrow T M_s,
\end{gather*}
when the reciprocal transformation $(e)$ is considered, the corresponding operator is
\begin{gather*}
 \Pi_e \colon \ T M_{(x,s)} \rightarrow T M_{(\bar x,\rho)}.
\end{gather*}
The B\"acklund chart \eqref{BC1*} allows to f\/ind solutions to initial boundary value problems of the Dym equation~\cite{Fuchssteiner:Carillo:1989a, BS4}. Furthermore, the links depicted in the B\"acklund chart relate each member of the involved hierarchies to the corresponding member of all the other hierarchies in the B\"acklund chart. Composition of B\"acklund transformations in~\eqref{BC1*} allows to retrieve and generalize to the whole hierarchies, the direct B\"acklund transformation between the KdV and Dym equations in~\cite{Guo:Rogers}.

The B\"acklund chart \eqref{BC1*} turned out to be very helpful in revealing new invariances enjoyed by any of the equations which appear in the B\"acklund chart itself. Indeed, let
\begin{gather*}
\{ \varphi ; x \} = \left( { \varphi_{xx} \over \varphi_x} \right)_x - {1 \over 2 }\left({ \varphi_{xx} \over \varphi_x} \right)^2
\end{gather*}
denote the {\it Schwarzian derivative}, then, the KdV Singularity equation $\varphi_t = \varphi_x \{ \varphi ; x\}$ is invariant under the M\"obius group of transformations
 \begin{gather*}%\label{mob}
 \widetilde\varphi={{a\varphi+b}\over{c\varphi+d}},\qquad a,b,c,d\in \mathbb{C} \qquad \text{such that} \quad ad-bc\neq 0,
\end{gather*}
since the Schwarzian derivative of $\varphi$ with respect to $x$, $ \{ \varphi ; x\}$, enjoys such a property.

The invariance, under the M\"obius group of transformations of all the members of the KdV Singularity hierarchy allowed to recover the {\it Invariance} $I$, enjoyed by the Dym hierarchy~\cite{Fuchssteiner:Carillo:1989a} and to f\/ind a new invariance enjoyed by the Kawamoto equation~\cite{Rogers:Carillo:1987b} in the case of the 5th order B\"acklund chart.

\subsection*{Acknowledgements}

The f\/inancial support of G.N.F.M.-I.N.d.A.M., I.N.F.N.\ and Sapienza University of Rome, Italy are gratefully acknowledged. C.~Schiebold wishes also to thank S.B.A.I.\ Dept.\ and Sapienza University of Rome for the kind hospitality. The authors wish to thank the referees who carefully read this work: their comments were of help in improving its presentation.

\pdfbookmark[1]{References}{ref}
\LastPageEnding

\end{document}